\documentclass[11pt]{amsart}
\usepackage[a4paper]{geometry}
\usepackage{amssymb, algorithm, algorithmic, tikz}
\usepackage[foot]{amsaddr}

\newcommand{\maj}{\mathop{\mathrm{maj}}}
\newcommand{\wt}{\mathop{\mathrm{wt}}}
\newcommand{\cut}{\mathop{\mathrm{cut}}}
\newcommand{\cubicmaxcut}{\textsc{Cubic\-MaxCut}}
\newcommand{\Nset}{\mathbb{N}}

\newcommand{\Ztilde}{\widetilde Z}
\newcommand{\NP}{\mathrm{NP}}
\newcommand{\RP}{\mathrm{RP}}

\newcommand{\bfz}{\mathbf{z}}

\newtheorem{theorem}{Theorem}
\newtheorem{lemma}[theorem]{Lemma}
\newtheorem{corollary}[theorem]{Corollary}

\let\phi=\varphi

\title[The antiferromagnetic Ising model beyond line graphs]{The antiferromagnetic Ising model\\ beyond line graphs}\thanks{Work partly supported by grant UKRI2771: Zeros, Algorithms, and Correlation for Graph Polynomials.}
\author{Mark Jerrum}
\address{School of Mathematical Sciences, Queen Mary, University of London, Mile End Road, London E1~4NS.}
\email{m.jerrum@qmul.ac.uk}

\begin{document}

\maketitle

\begin{abstract}
Both the antiferromagnetic Ising model and the hard-core model could be said to be tractable on line graphs of bounded degree.  For example, Glauber dynamics is rapidly mixing in both cases.  In the case of the hard-core model, we know that tractability extends further, to claw-free graphs and somewhat beyond.  In contrast, it is shown here that the corresponding extensions are not possible in the case of the antiferromagnetic Ising model.
\end{abstract}

\section{Motivation}

This note is a contribution to the study of the computational complexity of spin systems on restricted graph classes.  Spin systems as considered here are based on an underlying simple graph~$G$, and have as their `configurations' assignments of `spins' to the vertices $V(G)$ or edges $E(G)$ of~$G$.  In this introduction, concepts will be treated informally, but more precise definitions may be found on Section~\ref{sec:prelim}.  Computational questions about spin systems are invariably hard to answer in general, but may become easier when $G$ is restricted.  Examples include the classical exact solution of the Ising model on planar graphs, or the approximate computation of the partition function of the dimer model on bipartite graphs.  Restricting the graph $G$ also leads to fascinating open problems such as the complexity of approximating the partition function of the hard-core (independent set) model on bipartite graphs.

This note is about the antiferromagnetic Ising model.  However, for motivation, we need to start with another antiferromagnetic model, the hard-core model.  The configurations here are independent sets in~$G$, and the density of the configurations is controlled by a `fugacity' $\lambda>0$. See equation~\eqref{eq:hardcoredef}.

Generally, the hard-core model becomes intractable at high $\lambda$, when vertices in the independent set pack more tightly.   But restricted to so-called line graphs, the hard-core model becomes equivalent to the monomer-dimer (matchings) model.  Tractability of the monomer-dimer model expresses itself in various ways: Glauber dynamics mixes in polynomial time for all~$\lambda$ (Jerrum \& Sinclair \cite{JerrumSinclair}) and even in optimal $O(n\log n)$ time when $G$ has bounded degree (Chen, Liu and Vigoda~\cite{ChenLiuVigodaOptimal});  complex zeros of the partition function lie on the (negative) real axis (Heilmann and Lieb~\cite{HeilmannLieb});  and  there is a deterministic polynomial-time approximation algorithm for the partition function (Bayati et al.~\cite{BayatiEtAl}).  The last of these requires that~$G$ has bounded degree, and we shall \emph{assume throughout the rest of this note that this is the case};  more precisely, we consider graph families whose members have degree bounded by a universal constant~$\Delta$.   

In the case of the hard-core model we can range further in our search for tractability, specifically to claw-free graphs, a graph class that strictly includes line graphs.   Matthews showed that Glauber dynamics mixes in polynomial time~\cite{Matthews}, and Chen and Gu \cite{ChenGu} strengthened this to optimal $O(n\log n)$ time; Chudnovsky and Seymour~\cite{ChudnovskySeymour} showed that complex  zeros continue to lie on the real axis.  Patel and Regts~\cite{PatelRegts} presented a polynomial-time deterministic algorithm for the partition function, building on work of Barvinok~\cite{Barvinok}.  And then we can drive the search onwards to more exotic hereditary graph classes obtained by excluding subdivisions of claws, with tractability results from Bencs~\cite{BencsPhD}, Jerrum~\cite{JerrumHfree} and Jerrum and Patel~\cite{JerrumPatel}.

Turning attention finally to the antiferromagnetic Ising model, we find that line graphs are again tractable, with Glauber dynamics shown to mix in polynomial time by Dyer, Heinrich, Jerrum and M\"uller \cite{DyerHeinrichJerrumMuller}, a result that was improved to optimal $O(n\log n)$ time by Chen, Liu and Vigoda \cite{ChenLiuVigodaHolant}.  (The interaction strength may be arbitrarily large, but makes itself felt in the constant implicit in the O-notation.) Bencs, Csikv\'ari and Regts~\cite{BencsCsikvariRegts} show that zeros of the (univariate) partition function lie on the real axis, and deduce the existence of a polynomial-time deterministic algorithm for approximating the partition function. Given the similarity of the two models --- both having two spins and being antiferromagnetic --- it is natural to ask if we can take tractability of the the Ising model beyond line graphs. 

In the case of the hard-core model, the first step was to claw-free graphs.  The ground (minimum energy) configurations of the antiferrognetic Ising model correspond to maximum cuts in~$G$.  Bodlaender and Janson \cite{BodlaenderJanson} showed that finding a maximum cut is NP-complete in claw-free graphs, strongly suggesting that the antiferromagnetic Ising model is intractable on this graph class.\footnote{A further hint is provided by Bencs, Csikv\'ari and Regts~\cite[\S5]{BencsCsikvariRegts} who present claw-free graphs whose Ising partition function is not real-rooted.}    In this note we dot the i's and cross the t's to show that this is indeed the case. The methods employed are standard, dating back at least to Luby and Vigoda~\cite{LubyVigoda}.  Theorem~\ref{thm:NPhard} asserts that it is NP-hard to approximate the partition function, and that being able to sample configurations efficiently would imply $\RP=\NP$.  Theorem~\ref{thm:torpid} asserts that Glauber dynamics requires exponential time to mix in general, this being unconditional and independent of complexity-theoretic assumptions.

Although claw-free graphs were the initial goal of this note, it turns out that the constructions employed work for the smaller class of quasi-line graphs, and the even smaller class of line-graphs of multigraphs.  The latter is very close to line graphs, and suggests that line-graphs are the limit of tractability for the antiferromagnetic Ising model with strong interactions.  Further circumstantial evidence is provided by a survey of hereditary graph classes lying just above line~\cite{Graphclasses}, which suggests that maximum cardinality cut is NP-complete beyond line graphs.  However, in contrast to the hard-core model, where induced subdivided claws completely explain tractability (at least for hereditary classes excluding a finite set of induced subgraphs), a precise boundary between tractability and intractability is not established here.

\section{Preliminaries}
\label{sec:prelim}

Suppose $G$ is a graph, $\mu>1$ and $\lambda>0$.  Let $n=|V(G)|$ and $m=|E(G)|$.  We are concerned with two spin systems on~$G$.  In the first of these we shall take the spins to be $+1$ and $-1$, but write $+$ and $-$ for compactness.

For $\sigma:V(G)\to \{+,-\}$ the \emph{cut} defined by $\sigma$ is 
$$\cut\sigma=\big\{e\in E(G):\sigma(e)=\{+,-\}\big\}$$ 
and $\wt\sigma=\mu^{|\cut \sigma|}$ is its \emph{weight}.  Then define the (antiferromagnetic) \emph{Ising partition function} by 
$$Z(G,\mu)=\sum_{\sigma:V(G)\to\{+,-\}}\!\!\wt\sigma.$$
The partition function is the normalising factor for the corresponding \emph{Gibbs distribution}, which is given by $\sigma\mapsto\wt(\sigma)/Z(G,\mu)$ for all $\sigma:V(G)\to \{+,-\}$.

The results in this note will concern the partition function as $Z(G,\mu)$ as defined above; but to make contact with existing work, we need to consider a generalisation of $Z(G,\mu)$ that includes a contribution from individual vertices as well as edges.  To this end, introduce vertex weights $\bfz=\{z_v:v\in V(G)\}$ and define  
$$Z(G,\mu,\bfz)=\sum_{\sigma:V(G)\to\{+,-\}}\!\!\wt\sigma\prod_{v:\sigma(v)={+}}\!\!z_v.$$
The variable $z_v$ may be viewed as representing the effect of an `external field' acting on vertex~$v$.  The complex zeros of $Z(G,\mu,\bfz)$, regarded as a polynomial in~$\bfz$, have been extensively studied, and give insight into the phenomenon of phase transition in the Ising model.  There is also interest in the univariate polynomial obtained by substituting a single variable~$z$ for each of the~$z_v$, which corresponds physically to a uniform field.  Roughly speaking, the absence of zeros in the neighbourhood of the positive real axis is an indicator of the absence of a phase transition.

An \emph{independent set} in $G$ is a set of vertices $U\subseteq V(G)$ such that no edge in $G$ has both vertices in~$U$.  Denote by $\mathcal I(G)$ the collection of all independent sets in~$G$.  Then the independence polynomial, or \emph{hard-core partition function}, is defined by 
\begin{equation}\label{eq:hardcoredef}
Z_\mathrm{IS}(G,\lambda)= \sum_{U\in\mathcal{I}(G)}\lambda^{|U|}.
\end{equation}
As before, the partition function is the normalising factor for the Gibbs distribution on independent sets, which is given by $U\mapsto \lambda^{|U|}/Z_\mathrm{IS}(G,\lambda)$.
Again, there is value in studying the complex zeros of $Z_\mathrm{IS}(G,\lambda)$ regarded as a polynomial in~$\lambda$.  There is an obvious multivariate generalisation.  

For both of the above spin systems (and others) the main computational concerns are with approximating the partition function within (say) a constant factor, and with the mixing time (time to near equilibrium) of a simple Markov chain with single site updates knows as \emph{Glauber dynamics}.  Specialised to the Ising model, this is defined as in Algorithm~\ref{alg:Glauber}.

\begin{algorithm}
\begin{algorithmic}
\STATE {Suppose $X_t=\sigma$}
\COMMENT {Configuration at time $t$}
\STATE {Sample $v\in V(G)$ uniformly at random}
\STATE {$\sigma'(u)\leftarrow\sigma(u)$ for all $u\in V(G)\setminus\{v\}$}
\STATE {$\sigma'(v)\leftarrow-\sigma(v)$}
\STATE {With probability $\wt(\sigma')/(\wt(\sigma)+\wt(\sigma'))$ let $X_{t+1}\leftarrow\sigma'$}
\STATE {With probability $\wt(\sigma)/(\wt(\sigma)+\wt(\sigma'))$ let $X_{t+1}\leftarrow\sigma$}
\end{algorithmic}
\caption {Glauber dynamics for the Ising model on a graph $G$.}\label{alg:Glauber}
\end{algorithm}

The mixing time of an ergodic Markov chain $(X_n:n\in\Nset)$ on state space~$\Omega$ is the maximum over starting states $\omega\in\Omega$, of the minimum time $t\in\Nset$ at which the law of~$X_t$ conditioned on $X_0=\omega$ is within variation distance $\frac14$ of the limiting distribution.

The graph classes that are most relevant in this note are line, quasi-line and claw-free.  Suppose $H$ is a graph.  The line graph of $H$ is the graph $G=L(H)$ with vertex set $V(G)=E(H)$ and edge set
\begin{equation}\label{eq:line}
E(G)=\big\{\{e,f\}:e,f\in E(H), e\not=f\text{ and }|e\cap f|\not=\emptyset\big\}.
\end{equation}
A graph is \emph{line} if it is the line graph $L(H)$ of some base graph~$H$.  A graph is \emph{claw-free} if it does not contain a claw, i.e., the graph $K_{1,3}$, as an induced subgraph.  Finally a graph is \emph{quasi-line} if, for every vertex $v$, the neighbours of $v$ may be partitioned into two sets, both of which induce a clique.  (There may be additional edges in the neighbourhood of~$v$.)  It is clear that, in a line graph, the neighbourhood of every vertex can be so partitioned.  It is equally clear that no quasi-line graph contains a claw.  We therefore have 
$$
\{\text{line graphs}\}\subset \{\text{quasi-line graphs}\}\subset \{\text{claw-free graphs}\}.
$$
The inclusions are both strict.  All three graph classes are hereditary, i.e., closed under taking induced subgraphs. Both line graphs and claw-free graphs have been extensively studied, while quasi-line graphs have received less, but still significant, attention~\cite{ChudnovskySeymourQuasi}.  In the context of spin systems, notable line graphs include the kagome and pyrochlore lattices.

Finally, although the main theorems are stated for quasi-line graphs, they in fact apply to the even more restricted (and apparently more obscure) class of \emph{line graphs of multigraphs}.  Suppose $H$ is a multigraph, by which we understand a graph in which parallel edges, but not loops, are allowed.  Its line graph $G=L(H)$ has vertex set $V(G)=E(H)$ and edge set $E(G)$ exactly as defined in~\eqref{eq:line}.  Note that $G$ is a simple graph:  parallel edges $e,f$ in $H$ generate a single edge in $G$, even though they share two endpoints.  It is immediate that a line graph is, in particular, a line graph of a multigraph, and it is easy to check that a line graph of a multigraph is a quasi-line graph.  It may be helpful to think of a line graph as a graph that is formed as a (non-disjoint) union of cliques in which (i)~each vertex is in at most two cliques, and (ii)~each pair of cliques intersect in at most one vertex.  Then a line graph of a multigraph may be viewed as simply dropping condition~(ii). All inclusions mentioned are strict, so the complete picture is 
$$
\{\text{line graphs}\}\subset \{\text{line graphs of multigraphs}\}\subset\{\text{quasi-line graphs}\}\subset \{\text{claw-free graphs}\}.
$$

\section{Complexity-theoretic analysis}

A graph is \emph{cubic} if all its vertices have degree~3.  $\cubicmaxcut$ is the following problem.
\begin{description}
\item[Instance] A cubic graph $G$; 
\item[Output] $\max_{\sigma:V(G)\to\{+,-\}}|\cut\sigma|$.  
\end{description}
It is NP-hard.  More than that, Berman and Karpinski~\cite[Thm 1(1)]{BermanKarpinski} show:

\begin{lemma}\label{lem:BK}
It is $\NP$-hard to approximate $\cubicmaxcut$ within ratio $0.997$. That is to say, it is $\NP$-hard to compute a number $c$ such that $0.997C\leq c\leq C$, where $C=\cubicmaxcut(G)$. 
\end{lemma}
Berman and Karpinski state the result for graphs of maximum degree 3, but their reduction actually constructs a cubic graph (refer to \cite[\S3]{BermanKarpinski}).

Before we come to the proof of the main result, a small observation.
\begin{lemma}\label{lem:Cgen}
Every cubic graph $G$ has a cut of size at least $n=|V(G)|$.
\end{lemma}

\begin{proof}
Let $\sigma:V(G)\to\{+,-\}$ be a maximum cardinality cut in $G$.  If some vertex $v$ with $\sigma(v)={+}$ is adjacent to at least two vertices $u_1$, $u_2$ with $\sigma(u_1)=\sigma(u_2)={+}$ then we could obtain a larger cut by flipping the sign of~$v$.  It is similar when $\sigma(v)={-}$.  So every vertex in $G$ is adjacent to at least two vertices on the opposite side of the cut, and hence $|\cut\sigma|\geq n$.
\end{proof}

\begin{theorem}\label{thm:NPhard}
There is a universal constant $\mu>1$ such that, restricted to quasi-line graphs $G$ of degree at most~8:  
\begin{itemize}
\item the function $G\mapsto Z(G,\mu)$ is $\NP$-hard to approximate within factor 2, and
\item it is not possible to sample approximately from the Gibbs distribution of the Ising instance $(G,\mu)$ in polynomial time, unless $\mathrm{RP}=\mathrm{NP}$. 
\end{itemize}
``Approximately'' can be taken to mean realising an output distribution that is within variation distance $\frac14$ of the true one.
\end{theorem}

The constant 2 as an approximation factor is arbitrary, and could be replaced by any polynomial, or even modestly exponential function of~$n$.

\begin{proof}
Suppose $G$ is an instance of $\cubicmaxcut$. Let $n=|V(G)|$ and $m=|E(G)|$, so that $m=\frac32n$.

\tikzstyle{vertex} = [draw, shape=circle,minimum size=1mm,  inner sep=1pt, fill]
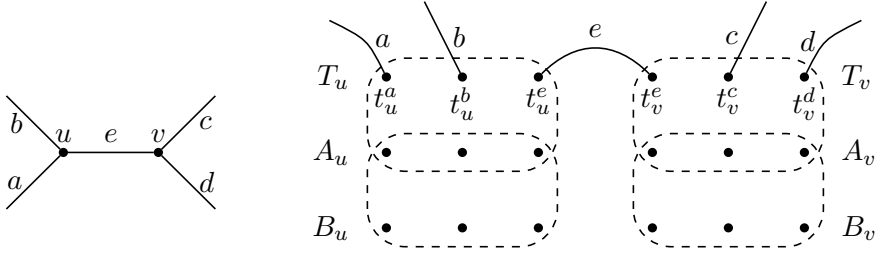
\begin{figure}
\begin {tikzpicture}[scale=0.5, semithick]

\draw [rounded corners = 5mm, dashed] (-0.5,-0.5) rectangle (4.5,2.5);
\draw [rounded corners = 5mm, dashed] (-0.5,1.5) rectangle (4.5,4.5);

\draw (0,0) node [vertex]  {} node [left=10pt] {$B_u$};
\draw (2,0) node [vertex]  {};
\draw (4,0) node [vertex]  {};
\draw (0,2) node [vertex]  {} node [left=10pt] {$A_u$};
\draw (2,2) node [vertex]  {};
\draw (4,2) node [vertex]  {};
\draw (0,4) node [vertex] (t1u) {} node [below] {$t_u^a$} node [left=10pt] {$T_u$};
\draw (2,4) node [vertex] (t2u) {} node [below] {$t_u^b$};
\draw (4,4) node [vertex] (t3u) {} node [below] {$t_u^e$};

\draw [rounded corners = 5mm, dashed] (6.5,-0.5) rectangle (11.5,2.5);
\draw [rounded corners = 5mm, dashed] (6.5,1.5) rectangle (11.5,4.5);

\draw (7,0) node [vertex]  {};
\draw (9,0) node [vertex]  {};
\draw (11,0) node [vertex]  {} node [right=10pt] {$B_v$};
\draw (7,2) node [vertex]  {};
\draw (9,2) node [vertex]  {};
\draw (11,2) node [vertex]  {}  node [right=10pt] {$A_v$};
\draw (7,4) node [vertex] (t1v) {} node [below] {$t_v^e$};
\draw (9,4) node [vertex] (t2v) {} node [below] {$t_v^c$};
\draw (11,4) node [vertex] (t3v) {} node [below] {$t_v^d$}  node [right=10pt] {$T_v$};

\draw (t3u) .. controls (5,5) and (6,5) .. node [above] {$e$} (t1v);
\draw (t1u) .. controls +(-0.5,1) .. node [right] {$a$} +(-1.5,1.5);
\draw (t2u) -- node [right] {$b$} +(-1,2);
\draw (t2v) -- node [left] {$c$} +(1,2);
\draw (t3v) .. controls +(0.5,1) .. node [left] {$d$} +(1.5,1.5);

\draw (-8.5,2) node [vertex] (u) {} node [above] {$u$};
\draw (-6,2) node [vertex] (v) {} node [above] {$v$};
\draw (v) to node [right] {$c$} +(1.5,1.5);
\draw (v) to node [right] {$d$} +(1.5,-1.5);
\draw (u) to node [left] {$b$} +(-1.5,1.5);
\draw (u) to node [left] {$a$} +(-1.5,-1.5);
\draw (u) to node [above] {$e$} (v);

\end{tikzpicture}
\caption{An edge $e=\{u,v\}$ in $G$ and its translation in the graph $G^*$.  The dashed `cartouches' indicate cliques.}
\label{fig:gadget}
\end{figure}

Construct an instance $(G^*,\mu)$ of the antiferromagnetic Ising model as follows.  First, for all $v\in V(G)$ define the graph $H_v$ on vertex set 
$$
V(H_v)=T_v\cup A_v\cup B_v
$$
where $T_v$, $A_v$ and $B_v$ are disjoint sets of cardinality~3.  Here, $T_v$ is a set of `terminals' which has extra structure:
$$
T_v=\{t_v^e:e\ni v\}.
$$
Thus, there is a terminal for every pair $v,e$, where $e$ is an edge incident at vertex~$v$.  The edge set $E(H_v)$ is obtained as a union of a clique on $T_v\cup A_v$ and a clique on $A_v\cup B_v$.  There are no further edges in $E(H_v)$.\footnote{There is a similarity between $H_v$ and one of the counterexamples given by Bencs, Csikv\'ari and Regts~\cite[\S5]{BencsCsikvariRegts}. In retrospect this is unsurprising, as to escape tractability we need to employ a (small) graph whose partition function has zeros that do not all lie on the real axis.} Now define $G^*$ by taking the disjoint union of the graphs $H_v$ over $v\in V(G)$ and adding some extra edges.  Specifically,
\begin{align*}
V(G^*)&=\bigcup_{v\in V(G)}V(H_v)\\
\noalign{\noindent and}
E(G^*)&=\big\{\{t_u^e,t_v^e\}:e=\{u,v\}\in E(G)\big\} \cup \bigcup_{v\in V(G)}E(H_v).
\end{align*}
The edges of $G^*$ that lie within some $H_v$ are internal, and the other edges external.  See Figure~\ref{fig:gadget}. 
We remark that $G^*$ is a quasi-line graph and has maximum degree 8.

Suppose the maximum cut in $G$ has cardinality~$C$.  We show that, for sufficiently large universal constant $\mu>1$, independent of~$n$, a sample from the Gibbs distribution on $(G^*,\mu)$ will, with high probability, point out a cut of size at least $0.997C$ in $G$.

Let $\sigma^*:V(G^*)\to\{+,-\}$ be an Ising configuration on $G^*$. The induced configuration $\phi(\sigma^*):V(G)\to\{+,-\}$ on~$G$ is defined by $\phi(\sigma^*)=\sigma$ where 
$$
\sigma(v)=\maj\{\sigma^*(t_v^e):e\ni v\},\quad \text{for all $v\in V(G)$},
$$
where $\maj$ denotes majority. Define 
$$
Z_\sigma=Z_\sigma(G^*,\mu)=\sum_{\sigma^*\in \phi^{-1}(\sigma)}\wt\sigma^*.
$$
so that 
$$
Z(G^*,\mu)=\sum_{\sigma:V(G)\to\{+,-\}}Z_\sigma.
$$
Suppose  that $\sigma$ specifies a cut of cardinality $c$ in $G$.  We claim that   
\begin{equation}\label{eq:Zsandwich}
\mu^{18n+c} \leq Z_\sigma(G^*,\mu)\leq 2^{8n}\mu^{18n+c}.
\end{equation}

Before verifying \eqref{eq:Zsandwich} we need to check the key property of the `gadget' graphs $\{H_v\}$.  Call a configuration $\sigma^*$ \emph{perfect on $H_v$} if 
$$
\sigma^*(T_v)=\{\pm\}, \>
\sigma^*(A_v)=\{\mp\}, \text{\ and\ }\> \sigma^*(B_v)=\{\pm\}.
$$
Observe that the (internal) edges of $H^v$ contribute $\mu^{18}$ to $\wt\sigma^*$ if $\sigma^*$ is perfect on~$H^v$, and at most $\mu^{17}$ otherwise.  To see this, consider the two cases (up to symmetry) $\sigma^*(A_v)=\{+\}$ and $\sigma^*(A_v)=\{+,-\}$.  Then count the extensions to $T_v$ and~$B_v$.

Call  $\sigma^*$ \emph{perfect} if it is perfect on $H_v$ for all $v\in V(G)$.  There is a unique perfect configuration $\sigma_0^*\in \phi^{-1}(\sigma)$.  This configuration has weight $\wt\sigma_0^*=\mu^{18n+c}$.  This deals with the lower bound in \eqref{eq:Zsandwich}

For the upper bound, let $\sigma^*_0$ be as before, and $\sigma^*\in\phi^{-1}(\sigma)$ be arbitrary.  Suppose we transform $\sigma^*$ to $\sigma^*_0$ by reassigning spins on each imperfect $H_v$, for each $v\in V(G)$ in turn.   The contribution to $\wt\sigma^*$ from internal edges in a previously imperfect $H_v$ is increased by a factor $\mu$ or larger.  However, it is possible that one external edge changes from $+-$ or $-+$ to $--$ or $++$, decreasing the weight by a factor~$\mu$.  In all cases, the weight does not decrease, so $\wt \sigma^*\leq\wt \sigma^*_0$.   The upper bound in \eqref{eq:Zsandwich} comes from noticing that $|\phi^{-1}(\sigma)|=2^{8n}$.

Now let $C$ be the cardinality of a maximum cut in $G$ and $0<c<C$.  Noting \eqref{eq:Zsandwich}, we have
\begin{equation}\label{eq:lbZ}
\sum_{\substack{\sigma:V(G)\to\{+,-\}\\ |\cut\sigma|>c}} Z_\sigma \geq 
\sum_{\substack{\sigma:V(G)\to\{+,-\}  \\ |\cut\sigma|=C}} Z_\sigma \geq 2\mu^{18n+C}.
\end{equation}
and 
\begin{equation}\label{eq:ubZ}
\sum_{\substack{\sigma:V(G)\to\{+,-\} \\ |\cut\sigma|\leq c}} Z_\sigma \leq 2^n\times  2^{8n}\mu^{18n+c}=2^{9n}\mu^{18n+c}.
\end{equation}

Set $c=0.997 C$, so that $[c,C]$ matches the gap specified in Lemma~\ref{lem:BK} and recall that $C\geq n$, by Lemma~\ref{lem:Cgen}.   Then 
\begin{equation}\label{eq:bvsb'}
\mu^{18n+C}\geq 2^{9n+2}\mu^{18n+c}
\end{equation} 
provided $n\geq10$ and we set $\mu\geq 2^{3067}$.  Thus, combining \eqref{eq:lbZ}, \eqref{eq:ubZ} and \eqref{eq:bvsb'}, a random configuration~$\sigma^*$ from the Gibbs distribution on~$G^*$ will satisfy $|\cut \phi(\sigma^*)|>c=0.997C$ with probability at least~$\frac34$.  Suppose we had a polynomial-time approximate sampler for the Gibbs distribution on bounded-degree quasi-line graphs that gets within variation distance~$\frac14$.  Then, via the mapping $\phi$ and with probability at least~$\frac12$, we could find a cut in~$G$ of cardinality at least~$c$.  By Lemma~\ref{lem:BK}, this would imply $\RP=\NP$.  

Now suppose we have an algorithm that produces an approximation $\Ztilde(G,\mu)$ to $Z(G,\mu)$ satisfying 
$$
\tfrac12 Z(G,\mu)\leq \Ztilde(G,\mu)\leq 2Z(G,\mu).
$$
Then, if $G$ has a maximum cut of cardinality $C$ we have 
\begin{align*}
\Ztilde(G,\mu)&\geq \tfrac12 Z(G,\mu) \geq \mu^{18n+C},&\text{(by \eqref{eq:lbZ})}\\
\noalign{\noindent whereas, if $G$ has no cuts of cardinality greater than $c$, then}
\Ztilde(G,\mu)&\leq 2 Z(G,\mu) \leq 2^{9n+1}\mu^{18n+c}\leq \tfrac12\mu^{18n+C}.&\text{(by \eqref{eq:ubZ} and \eqref{eq:bvsb'})}
\end{align*}
We can clearly distinguish these two situations.  In light of Lemma~\ref{lem:BK} we have a polynomial-time reduction from approximating $\cubicmaxcut(G)$ to approximating the partition function of the Ising model on a bounded-degree quasi-line graph~$G^*$. 
\end{proof}

Two remarks.    The first, obvious one is that the above result places a (super) astronomic bound on $\mu$, and clearly says nothing about computation in this universe.  Although there is scope for tightening the estimates, it is unlikely, given the starting point in Lemma~\ref{lem:BK}, that the bound can be reduced to merely astronomic.  However, it does establish that the positive results in \cite{ChenGu, JerrumHfree, JerrumPatel} relating to the hard-core (independent set) model do not have analogues for the antiferromagnetic Ising model. The second is that the graph $G^*$ constructed in the proof is special kind of quasi-line graph in which the cliques may be chosen in advance.  Another way of looking at this is that $G^*$ is a line graph of a multigraph.  

Readers familiar with the Asano contraction approach as described by Lebowitz, Ruelle and Speer~\cite{LebowitzRuelleSpeer}, might wonder why it does not yield a zero-free region of the partition function containing the real axis.  If it did, it would contradict Theorem~\ref{thm:NPhard} (which relies on the existence of multiple phases), or at least establish $\mathrm{P}=\mathrm{NP}$.  Consider a line graph as being constructed from cliques by identifying vertices. The roots of the partition functions of the individual cliques are known to all lie on the negative real axis.  The line graph itself then has a parabolic multivariate\footnote{We noted earlier that a larger univariate region has been obtained via a different approach~\cite{BencsCsikvariRegts}.} zero-free region containing the positive real axis~\cite{LebowitzRuelleSpeer,ChenLiuVigodaHolant}.  This works for line graphs proper, but not for quasi-line graphs.  The catch is hinted at in \cite[Ex.~2.2(c)]{LebowitzRuelleSpeer}:  while identifying vertices, parallel edges are created, and the resulting Ising system has non-constant interaction strengths.  Incidentally, this issue does not arise in the context of the hard-core model, where parallel edges are equivalent to single ones.  This to a modest extent explains the difference in properties between the two models.

\section{Exponential-time mixing}
Theorem~\ref{thm:NPhard} rules out any polynomial-time algorithm for approximating the partition function or approximately sampling configurations, under the assumption $\mathrm{NP}\not=\mathrm{P}$.  We also have an \emph{unconditional} exponential bound on mixing  time for Glauber (single-site) dynamics.  

An $n$-vertex graph~$G$ is a \emph{$c$-magnifier} if, for every vertex subset $U\subset V(G)$ with $|U|\leq\frac12|V(G)|$, it is the case that $|N_G(U)\setminus U|\geq c|U|$, where $N_G(U)$ denotes the neighbourhood of~$U$.

\begin{theorem}\label{thm:torpid}
There is a universal constant $\mu>1$, and a family $(G_n:n\in\Nset)$ of quasi-line graphs of degree at most~8, with $|V(G_n)|=18n$, such that the mixing time of Glauber dynamics on the Ising model $(G_n,\mu)$ is $\Omega(2^n)$.  
\end{theorem}

\begin{proof}
For each $n\in\Nset$ let $G_n$ be a cubic bipartite $n+n$ vertex graph that is an $\frac18$-magnifier.  The existence of such a graph follows from a lemma of Bassalygo \cite{Bassalygo}; this lemma is quoted by Alon, who explicitly gives the random bipartite cubic graph as a numerical example. See~\cite[Lemma 4.1]{Alon} and the discussion that follows.   Let $G_n^*$ be the graph obtained from $G_n$ by the construction used in the proof of Theorem~\ref{thm:NPhard}.  Note that $G_n^*$ is a quasi-line graph of maximum degree~8 with $2n\times9=18n$ vertices.  

Let $V(G_n)=L(G_n)\cup R(G_n)$ be the bipartition of $G_n$, and $\Omega_=$ (respectively $\Omega_+$, $\Omega_-$) be the set of cuts in $G_n$ in which the number of $+$ vertices in $L(G_n)$ is equal to (resepctively greater than, less than) the number  in $R(G_n)$.  Recall the function~$\phi$ from configurations in $G_n^*$ to configurations in $G_n$, used in the proof of Theorem~\ref{thm:NPhard}.  In transitioning from a configuration in $\phi^{-1}(\Omega_+)$ to one in $\phi^{-1}(\Omega_-$), Glauber dynamics on $G_n^*$ must pass through some configuration in $\phi^{-1}(\Omega_=)$.     

Suppose $\sigma\in\Omega_=$, and let 
$$S=(L(G_n)\cap \sigma^{-1}\{-\})\cup (R(G_n)\cap \sigma^{-1}\{+\})$$
We have $|S|=|S^c|=n$ where $S^c=V(G_n)\setminus S$.  Since $G_n$ is a $\frac18$ magnifier, there are at least $\frac18\times2n=\frac14n$ vertices in $S^c$ adjacent to some vertex in $S$, from which it follows that  $|\cut \sigma|\leq m-\frac14n$, where $m=3n$ is the number of edges in~$G_n$.  As in the previous proof, this entails $Z_\sigma\leq (2^8\mu^{18})^{2n}\mu^{m-n/4}$ and 
$$
\sum_{\sigma\in\Omega_=}Z_\sigma\leq 2^{2n}(2^8\mu^{18})^{2n}\mu^{m-n/4}=2^{18n}\mu^{36n}\mu^{m-n/4}.
$$
On the other hand ,
$$
\sum_{\sigma\in\Omega_\pm} Z_\sigma\geq (\mu^{18})^{2n}\mu^m=\mu^{36n}\mu^m.
$$
and hence
$$
\sum_{\sigma\in\Omega_=}Z_\sigma\leq 2^{-n}\sum_{\sigma\in\Omega_\pm}Z_\sigma,
$$
provided $\mu\geq2^{76}$.

This inequality implies that the `conductance' of Glauber dynamics on $G^*_n$ is exponentially small, and hence that the mixing time is exponentially large.  Specifically, assuming without loss of generality, that $\sum_{\sigma\in\Omega_+}Z_\sigma\leq\frac12$, we may set $A=\phi^{-1}(\Omega_+)$ and $M=\phi^{-1}(\Omega_=)$ in~\cite[Claim 7.14]{ETH}, to obtain a lower bound of $2^{n-2}$ on mixing time. 
\end{proof}

A couple of remarks.  Basically the same argument would work for any `cautious' Markov chain that changes only a small number of spins at each transition. Also, by Cheeger's inequality, the spectral gap of the Glauber dynamics on the Ising system $(G_n,\mu)$ is $O(2^{-n})$; equivalently, the relaxation time is $\Omega(2^n)$.

\section{Interactions of varying strength}
A feature of the existing positive results on the antiferromagnetic Ising model on line graphs \cite{DyerHeinrichJerrumMuller,ChenLiuVigodaHolant,BencsCsikvariRegts} is that they rely on constant strength interactions over all edges.  Suppose, instead of having a single interaction strength $\mu$, as previously, we allow two: $\mu$ (large) and 1.  Then the edges with weight 1 in a graph~$G$ have no influence on the weight $\wt\sigma$ of a configuration~$\sigma$ and could as well be removed from the instance~$G$ when considering the partition function.  Thus, computing $Z(G,\mu)$ for any $n$-vertex graph $G$ is the same as computing the partition function of the complete graph $K_n$ with a certain pattern of interaction strengths chosen from $\{\mu,1\}$. But $K_n$ is a line graph of the star $K_{n,1}$. So if we allow varying interactions, line graphs as as hard as arbitrary graphs. 

So it seems that the case of line graphs with varying interaction strengths is without interest, even if we allow just two strengths (admittedly, widely separated).  However, for us it is not \emph{quite} vacuous as we are exclusively concerned with classes of graphs of bounded degree, and it is apparently not the case that every family of bounded-degree graphs can be obtained as subgraphs of line graphs of bounded degree.

However, it is fairly easy to show that allowing two interaction strengths leads to intractability, even when vertex degrees are bounded.  Given a cubic graph~$G$, introduce graph gadgets $J_v$, with $v\in V(G)$, that are even simpler than before.  The vertex set of~$J_v$ is $V(J_v)=T_v\cup A_v$, where $T_v$ and~$A_v$ are sets of cardinality~3, and $T_v=\{t_v^e:e\ni v\}$.  The edge set of~$J_v$ is that of a compete bipartite graph with bipartition $T_v+A_v$. Form $G^\dag$ from~$G$ as we did before for~$G^*$, but using $J_v$ in place of~$H_v$.   Note that $G^\dag$ is a subgraph of a line graph of maximum degree~6.  

As before, we say that a configuration $\sigma^\dag$ on $G^\dag$ is perfect on $J_v$ if $\sigma^\dag(T_v)=\{\pm\}$ and $\sigma^\dag(A_v)=\{\mp\}$.  The internal edges in $J_v$ contribute a factor  $\mu^9$ to $\wt\sigma^\dag$ if $\sigma^\dag$ is perfect on $J_v$, and at most $\mu^6$ otherwise.  So the proofs of  Theorems \ref{thm:NPhard} and~\ref{thm:torpid} carry across from quasi-line graphs to subgraphs of line graphs (with an improved bound on $\mu$). 

\begin{corollary}
Theorems \ref{thm:NPhard} and~\ref{thm:torpid} hold with ``quasi-line graphs of  degree at most~8'' replaced by ``subgraphs of line graphs of maximum degree at most~6''.
\end{corollary}

It is possible to vary the interaction strengths within a restricted range and preserve tractability, but large variations cannot be tolerated.

\bibliographystyle{plain}
\bibliography{IsingClawFree}

\end{document}